\documentclass[a4paper,12pt]{article}
\usepackage{amsthm}
\usepackage{amssymb}
\usepackage{amsmath}
\usepackage{amsfonts}
\usepackage{mathrsfs}
\usepackage{bbm}
\usepackage{graphicx}
\usepackage[all,dvips,arc,frame]{xy}
\usepackage{float}

\newcommand{\Ad}{\mathrm{Ad}}
\newcommand{\Ker}{\mathrm{Ker}}
\newcommand{\Dom}{\mathrm{Dom}}

\newcommand{\Img}{\mathrm{Im}}

\newcommand{\beq}{\begin{equation}}
\newcommand{\eeq}{\end{equation}}
\newcommand{\beqn}{\begin{eqnarray}}
\newcommand{\eeqn}{\end{eqnarray}}
\newtheorem{lemma}{Lemma}
\newtheorem{theo}{Theorem}
\theoremstyle{remark}

\title{}

\def\G{\Gamma}

\def\U{\Upsilon}

\def\d{\partial}

\def\be{\begin{equation}}
\def\ee{\end{equation}}
\def\bea{\begin{eqnarray}}
\def\eea{\end{eqnarray}}
\def\A{{\cal A}}

\def\D{{\cal D}}

\def\G{{\cal G}}

\def\R{{\cal R}}
\def\S{{\cal S}}

\def\U{{\cal U}}

\def\F{{\cal F}}

\def\si{\sigma}

\def\L{\Lambda}

\begin{document}

\vspace*{0.5cm}
\begin{center}
{\Large \bf Dressing cosets revisited}

\end{center}

\vspace{0.3cm}

\begin{center}

 R. Squellari${}^{a,b}$ \\

\bigskip

{\it ${}^a$Institut de math\'ematiques de Luminy,
 \\ 163, Avenue de Luminy, \\ 13288 Marseille, France}
 
 \medskip
  
 ${}^b$ \it Laboratoire de Physique Th\'eorique et des
Hautes Energies\\
{\it CNRS, Unit\'e associ\'ee URA 280}\\
{\it 2 Place Jussieu, F-75251 Paris Cedex 05, France}

\bigskip

\end{center}

\vspace{0.2cm}

\vskip 2.5truecm

\begin{abstract}
We present an alternative algebraic derivation   of the  dual pair of nonlinear $\sigma$-models  based on the 'dressing cosets' extension
of the Poisson-Lie $T$-duality \cite{KS1}. Then we generalize the result to dual pairs of Lagrangians not considered in \cite{KS1}. Our generalization turns out to incorporate the dualisable models constructed  by Sfetsos  in \cite{Sfet1}.
\end{abstract}

\newpage

\section{Introduction}
  Poisson-Lie $T$-duality  \cite{KS2}   establishes a dynamical equivalence    of  certain nonlinear  $\sigma$-models,  the target manifolds of which are a Poisson-Lie group $G$ and its dual Poisson-Lie group $\tilde G$, respectively.  Those models admit a generalization \cite{KS1}
 for which the targets of the mutually dual $\sigma$-models are respectively the spaces of the dressing orbits $F\backslash G$ and $F\backslash \tilde G$ where $F$ is certain (isotropic) subgroup of the common Drinfeld double $D$ of $G$ and $\tilde G$.  The most 
 economic way to describe the Lagrangian dynamics of the models on $F\backslash G$ and $F\backslash \tilde G$
 is in terms of $\sigma$-models living  apparently on the targets $G$ and $\tilde G$ but admitting a local dressing
  gauge symmetry with the gauge group being the current group with values in $F$ \cite{KS1}.    We remark that this description is the dressing
  analogue of the 
  usual description of the Lagrangian dynamics of the standard  gauged $G/H$ WZW $\sigma$-models in terms of the $G$ target only (cf.  Eq. (2.8) of \cite{GK}).  In that case the gauge group
  $H$ acts  on the apparent target $G$ in the adjoint way.
  
  \medskip
  
  \noindent The brief   and somewhat sketchy derivation of the dynamics of the dressing cosets presented in \cite{KS1} was based on the methods of symplectic geometry namely on the  symplectic reduction of certain coadjoint orbit of the centrally extended loop group of 
  the Drinfeld double $D$.     In this article,
  we furnish a more algebraic derivation of the second-order action of the dressing cosets found in \cite{KS1}. Futhermore, we show how our new algebraic derivation leads to a generalization of the dressing cosets construction of Ref.\cite{KS1}. We identify explicitly the actions of the generalized dressing cosets and confirm by direct calculation their gauge invariance. Finally we show that the models constructed by Sfetsos in \cite{Sfet1} fit in our generalized construction.
   \medskip
   
   \noindent In Section 2, we introduce our new algebraic calculus through the example of the standard Poisson-Lie T-duality and, in Section 3, we use it to derive the actions of dressing coset models of Ref.\cite{KS1}. In Section 4, by still using the same algebraic method, we obtain new dual pairs of models which were not considered in \cite{KS1} and we call them the generalized dressing cosets. In Section 5, we check explicitly the gauge symmetry of all dressing cosets models (standard or generalized), and in Section 6 we show that the $\sigma$-models found by Sfetsos \cite{Sfet1} can be interpreted as the generalized dressing cosets. We end up with a short Section 7 containing the conclusions.

\section{Preliminaries}

Let us introduce few preliminary concepts which will be common for all particular cases studied in this article. Recall that a Drinfeld double  is  a $2n$-dimensional Lie group $D$ such that its  Lie algebra $\mathcal{D}$  contains 
 two    $n$-dimensional subalgebras $\mathcal{G}$ and $\hat{\mathcal{G}}$  which are complementary, i.e. $\mathcal{G}\cap\hat{\mathcal{G}}=\{0\}$, and which  are   isotropic with respect to a non-degenerate invariant bilinear form on $\mathcal{D}$ denoted $\langle.|.\rangle$, i.e. $\langle \mathcal{G}|\mathcal{G}\rangle=\langle\hat{\mathcal{G}}|\hat{\mathcal{G}}\rangle=0$. In this paper
 we shall consider only the so-called perfect Drinfeld doubles  for which {\it every}  element $l\in D$ can be written in a {\it unique} way
 as $l=g\hat h$ for some $g\in G$ and some  $\hat h\in\hat G$ where $G$ and $\hat G$ are the subgroups of $D$ which integrate
 the Lie subalgebras $\mathcal{G}$ and $\hat{\mathcal{G}}$ of $\D$.\\
The crucial structural ingredient of all models presented in this paper will be certain linear operator $R:\hat{\G}\to\G$ and its adjoint $R^*:\hat{\G}\to\G$ with respect to the bilinear form $\langle.|.\rangle$. In fact, we shall see that the properties of the kernels of $R$ and $R^*$ will classify the different classes of Poisson-Lie $\sigma$-models. Indeed, the case $\Ker R=\Ker R^*=\{0\}$ and $\Ker(R+R^*)=\{0\}$ underlies the standard Poisson-Lie T-duality considered in \cite{KS2}; the case $\Ker R=\Ker R^*=\{0\}$ and $\Ker (R+R^*)\neq\{0\}$ will lead to the dressing cosets introduced in \cite{KS1}; the general case is obtained when $\Ker R=\Ker R^{*}\neq\{0\}$ while $Ker(R+R^{*})\neq\{0\}$. The subcase of general case for which $\Ker R=\Ker R^*=\Ker(R+R^*)\neq\{0\}$ will give the models introduced by Sfetsos \cite{Sfet1}. \medskip 
\\
As a warm-up to help us to set up the language, we consider the simplest case of the standard Poisson-Lie T-duality: models with $\Ker R=\Ker R^*=\Ker (R+R^*)=\{0\}$.
\paragraph{First order action:} Parametrize a two-dimensional world-sheet by a space-like coordinate $\si$ and a time-like coordinate $\tau$  and
consider a  $D$-valued field $l(\si,\tau)$. For the first order action governing the dynamics
of these fields  we take 
\beq \S(l)=\int \big(\frac{1}{2}\langle\partial_\tau ll^{-1}|\partial_\sigma ll^{-1}\rangle+\frac{1}{12}d^{-1}\langle dll^{-1}\wedge[dll^{-1}\wedge dll^{-1}]\rangle+\frac{1}{2} K(\partial_\sigma ll^{-1})\big),
\label{action1}
\eeq
where we omit the measure $ d\tau d\sigma$ of integration for the rest of the article. In the action \eqref{action1},  $dll^{-1}=\partial_\tau ll^{-1}d\tau+\partial_\sigma ll^{-1}d\sigma$ denotes the pull-back of the right-invariant Maurer-Cartan form from $D$ to the world-sheet, and $K$ is a certain quadratic form defined on $\D$. Before expliciting the expression of the quadratic form $K$, we need the following theorem:
\begin{theo}
Let $R:\hat{\G}\to\G$ and $R^*:\hat{\G}\to\G$ such that $\Ker R=\Ker R^*=\{0\}$ and $\Ker(R+R^*)=\{0\}$, then
 there exists an unique pair of linear operators $\alpha_\pm:\D\to\hat{\G}$ such that any $X\in \D$ can be decomposed as:
\beq
X=\alpha_+(X)+\alpha_-(X)+R\alpha_+(X)-R^*\alpha_-(X).
\label{decomposition1}
\eeq
\end{theo}
\begin{proof}
Let us take $X\in \D$, then we know that there exists an unique decomposition $X=x+a$ where $x\in\G$ and $a\in\hat{\G}$. In our analyse we suppose the relation \eqref{decomposition1} verified, then by projecting on the algebras $\hat{\G}$ and $\G$ we obtain the system:
\beq
\left\{ \begin{array}{l}a=\alpha_{+}(X)+\alpha_{-}(X) \\ x=R\alpha_{+}(X)-R^{*}\alpha_{-}(X) \end{array}\right. , \;\textrm{or}
\left\{ \begin{array}{l}(R+R^*)\alpha_-(X)=Ra-x \\ (R+R^*)\alpha_+(X)=R^*a+x \end{array} \right. .
\eeq 
Since in our case $\Ker(R+R^*)$ is trivial, we find explicitly the expression of the $\alpha_\pm$ operators:
\beq
\left\{ \begin{array}{l} \alpha_-(X)=(R+R^*)^{-1}(Ra-x) \\ \alpha_+(X)=(R+R^*)^{-1}(R^*a+x) \end{array} \right. .
\label{systeme1}
\eeq
Thus we have found the unique pair of operators $\alpha_\pm$, which fit in the decomposition \eqref{decomposition1}.
\end{proof}
\noindent 
We define a quadratic form $K$ on $\D$ by:
\beq
K(X)=2\langle \alpha_+(X)|R\alpha_+(X)\rangle+2\langle \alpha_-(X)|R\alpha_-(X)\rangle .
\label{hamiltonien}
\eeq
This definition completes the description of the first order action \eqref{action1}. However, it is perhaps not evident at the first sight that the quadratic form $K$ coincides with the quadratic form appearing in Ref.\cite{KS2}, \cite{KS3}. Let us see that this is indeed true. We observe that the subspaces $\R=\{\gamma+R\gamma,\gamma\in\hat{\G}\}$ and its orthogonal complement $\R^\bot=\{\gamma-R^*\gamma,\gamma\in\hat{\G}\}$ are $n$-dimensional. Moreover we see that $\Ker(R+R^*)=\{0\}$ implies $\R\cap\R^\bot=\{0\}$, hence $\R\oplus\R^\bot=\D$. Recall that the quadratic form $K'$ appearing in Ref.\cite{KS3} has two eigenvalues $\pm1$ corresponding to the eigenspaces $\R$ and $\R^\bot$ respectively.\\
The result is that we obtain the expression of $K'(X)$ for all $X\in \D$ as:
\beqn
K'(X)&=&\langle \alpha_+(X)+R\alpha_+(X)|\alpha_+(X)+R\alpha_+(X)\rangle \nonumber \\
&& \phantom{rrrrrrrrrrrr}-\langle \alpha_-(X)-R^*\alpha_-(X)|\alpha_+(X)-R^*\alpha_-(X)\rangle \nonumber \\
&=&2\langle \alpha_+(X)|R\alpha_+(X)\rangle+2\langle \alpha_-(X)|R\alpha_-(X)\rangle =K(X).
\eeqn
We can easily determine the equation of motion verified by the $l(\si,\tau)\in D$ field, thus the extremalization of the action \eqref{action1} gives:
\beq
\partial_\tau ll^{-1} +\alpha_+(\partial_\sigma ll^{-1})-\alpha_-(\partial_\sigma ll^{-1})+R\alpha_+(\partial_\sigma ll^{-1})+R^*\alpha_-(\partial_\sigma ll^{-1})=0.
\label{motion1}
\eeq
\paragraph{Poisson-Lie $\si$-models:} Because the double $D$ is perfect, we can represent any field $l(\si,\tau)\in D$ as $l(\si,\tau)=g(\si,\tau)\hat{h}(\si,\tau)$ with $g(\si,\tau)\in G$ and $\hat{h}(\si,\tau)\in\hat{G}$. Using this parametrization and the Polyakov-Wiegmann formula \cite{PW} the action \eqref{action1} becomes:
\beq
\mathcal{S}(g,\hat{h})=\int\langle g^{-1}\partial_\tau g|\partial_\sigma \hat{h} \hat{h}^{-1}\rangle+\frac{1}{2}K(\partial_\sigma gg^{-1}+\Ad_g\partial_\sigma \hat{h} \hat{h}^{-1}) .
 \label{lag1}
\eeq
Applying the decomposition \eqref{decomposition1} on $X=\partial_\sigma l l^{-1}=\partial_\sigma gg^{-1}+\Ad_g\partial_\sigma \hat{h} \hat{h}^{-1}$, the action \eqref{lag1} becomes:
\beq
\mathcal{S}(g,\hat{h})=\int\langle \partial_\tau gg^{-1}| \alpha_++\alpha_-\rangle+\langle\alpha_+|R\alpha_+\rangle+\langle \alpha_-|R\alpha_- \rangle,
\label{lag2}
\eeq
where we denoted $\alpha_\pm(\partial_\sigma l l^{-1})=\alpha_\pm$ and used Eq \eqref{hamiltonien}. Now we define two operators $J(\hat{J})$ on $\mathcal{D}$, which project to $\mathcal{G}(\hat{\mathcal{G}})$ with the kernel $\hat{\mathcal{G}}(\mathcal{G})$ respectively. Let us project the decomposition:
\beq
\partial_\sigma l l^{-1}=\partial_\sigma gg^{-1}+\Ad_g\partial_\sigma \hat{h} \hat{h}^{-1}=\alpha_++\alpha_-+R\alpha_+-R^*\alpha_-,
\eeq
on the subspaces $\hat{\G}$ and $\G$:
\beqn
\alpha_++\alpha_-&=&\hat{J}\Ad_g\hat{J}\partial_\sigma \hat{h} \hat{h}^{-1} \label{eq_pi1}\\
R\alpha_+-R^*\alpha_-&=&\partial_\sigma gg^{-1}+J\Ad_g\hat{J}\partial_\sigma \hat{h} \hat{h}^{-1}.
\label{eq_pi2}
\eeqn
In Eq.\eqref{eq_pi1}, the operator $\hat{J}\Ad_g\hat{J}:\hat{\G}\to\hat{\G}$ is invertible, which gives:
\beq
\partial_\sigma \hat{h} \hat{h}^{-1}=\hat{J}\Ad_{g^{-1}}\hat{J}(\alpha_++\alpha_-).
\label{eq_pi3}
\eeq
Inserting relation \eqref{eq_pi3} in Eq.\eqref{eq_pi2}, we obtain:
\beqn
R\alpha_+-R^*\alpha_-&=&\partial_\sigma gg^{-1}+\Pi^{g*}(\alpha_++\alpha_-)\quad \textrm{or, equivalently:}\\
\partial_\sigma gg^{-1}&=&(R+\Pi^g)\alpha_+-(R+\Pi^g)^*\alpha_-.
\label{relation1}
\eeqn
Here 
\beqn
&&\Pi^g=J\Ad_g J\Ad_{g^{-1}}\hat{J}\; \textrm{ and}\nonumber \\
&&\Pi^{g*}=J\Ad_g \hat{J}\Ad_{g^{-1}}\hat{J}, \nonumber
\eeqn
are two $g$-dependent operators from $\hat{\G}$ into $\G$. By the way, it turns out that $\Pi^g$ defines the so called Poisson-Lie structure\footnote{ More precisely, the Poisson bracket on the group $G$ can be expressed in term of the operator $\Pi^g$ as $\{f,g\}=\langle \nabla f| \Pi^g \nabla g \rangle$. Here $(f,g)$ are two arbitrary fonctions of $G$ and $\nabla\in\hat{\G}$ is defined by $\langle \nabla |x\rangle=\nabla^{x}$ where $\nabla^{x}$ is a differential operator generating the left action of $\mathcal{G}\;(x\in\G)$ on $G$.} on $G$.
\medskip
\\
Consider now any skew symetric $g$ dependent operator $M^g:\hat{\G}\to\G$, i.e. verifying $\langle M^gx|y\rangle=-\langle x|M^gy\rangle,\; \forall (x,y) \in \hat{\G}$. The action \eqref{lag2} can be rewritten in a more complicated  but usefull way as:
\beq
\mathcal{S}(g,\hat{h})=\int\langle \partial_\tau gg^{-1}| \alpha_++\alpha_-\rangle+\langle\alpha_+|(R+M^g)\alpha_+\rangle+\langle \alpha_-|(R+M^g)\alpha_- \rangle.
\label{lag3}
\eeq
If moreover $R+M^g$ is invertible, then the action\eqref{lag3} can be written in even more complicated way as:
\beq
S(g,\hat{h})=-\frac{1}{2}\int \langle Y|(R+M^g)Y\rangle+\langle Z-\d_\tau gg^{-1}|(R+M^g)^{-1}(Z+\d_\tau gg^{-1}\rangle,
\label{lag3bis}
\eeq
where
\beqn
Y&=&(R+M^{g})^{-1}\bigg[(R+M^{g})\alpha_{+}+(R+M^{g})^{*}\alpha_{-}-\d_{\si}gg^{-1}\bigg]\\
Y&=&2\alpha_+-(R+M^g)^{-1}\d_+ gg^{-1}\\
Z&=&(R+M^g)\alpha_+-(R+M^g)^*\alpha_-.
\eeqn
Here $\xi_\pm=\tau\pm\si$ and $\d_\pm=\d_\tau\pm\d_\si$ are the light-cone coordinates.\\
Why this apparent complication ? Because, we want to study the variation problem, in which $\hat{h}$ varies while $g$ remains fixed. Indeed, it turns out that there is a choice of the skew adjoint operator $M^g$ for which the $Z$ part of the action  \eqref{lag3bis} depends only on $g$, and all dependence on $\hat{h}$ is hidden only in the quantity $Y$. This choice is simply $M^g=\Pi^g$, because of the relation \eqref{relation1}, we see that $Z=\d_\si gg^{-1}$ and is $\hat{h}$ independent. Futhermore, the following relation (Eq.\eqref{relationY}) between $Y$ and $\d_\si \hat{h}\hat{h}^{-1}$ becomes invertible (i.e. one can express $\d_\si\hat{h}\hat{h}^{-1}$ in term of $Y$) since it only involves the invertible operators  $(R+\Pi^g)$, $(R+R^*)$ and $\hat{J}\Ad_g\hat{J} $. The relation relating the $Y$ variable to $\d_\si \hat{h}\hat{h}^{-1}$, is given by:
\beqn
Y&=&2(R+R^*)^{-1}(R+\Pi^g)^*\hat{J}\Ad_g\hat{J}\d_\si \hat{h}\hat{h}^{-1}\nonumber \\
&&+2(R+R^*)^{-1}\d_\si gg^{-1}-(R+\Pi^g)^{-1}\d_+gg^{-1},  
\label{relationY}
\eeqn
Thus in the minimization of the action \eqref{lag3bis} with $M^g=\Pi^g$, the variation of the $n$ components of the variable $\d_\si \hat{h}\hat{h}^{-1}$ is equivalent to the variation of the $n$ variables in $Y$.\\
Consequently, in the action \eqref{lag3bis}  the quadratic term containing $Y$ simply disappears upon the variation since there is no linear $Y$ term. Upon replacing $Z$ by $\d_\si gg^{-1}$ we conclude  that the $\hat{h}$-variation leads to the standard Poisson-Lie second order action\cite{KS2}:
\beq
\S(g)=\frac{1}{2}\int d\xi_+d\xi_-\langle \partial_-gg^{-1}|(R+\Pi^g)^{-1}\partial_+gg^{-1}\rangle,\quad g \in G.
\label{action2}
\eeq
\paragraph{Duality:}
Where is the duality in all the formalism which we have just presented ? To see it, we define a linear operator $\hat{R}:\G\to\hat{\G}$ given by
\beq \hat{R}=R^{-1},
\label{duality1}
\eeq
and consider its adjoint $\hat{R}^*:\G\to\hat{\G}$. Obviously Eq.\eqref{duality1} implies that $\Ker \hat{R}=\Ker \hat{R}^*=\{0\}$ and $\Ker (\hat{R}+\hat{R}^*)=\{0\}$. From the Theorem 1, we conclude that there exists an unique set of linear operators $\hat{\alpha}_\pm:\D\to\G$ such that for any $X\in\D$ we can decompose it as:
\beq
X=\hat{\alpha}_+(X)+\hat{\alpha}_-(X)+\hat{R}\hat{\alpha}_+(X)-\hat{R}^*\hat{\alpha}_-(X).
\eeq
Now we can easily verify that the quadratic form $K$ on $\D$ defined by \eqref{hamiltonien} can be rewritten as:
\beq
K(X)=2\langle \hat{\alpha}_+|\hat{R}\hat{\alpha}_+\rangle+2\langle \hat{\alpha}_-|\hat{R}\hat{\alpha}_-\rangle.
\label{dual_quadratic}
\eeq
Indeed, Eq.\eqref{dual_quadratic} comes from the fact that the subspaces $\R$ and $\R^\bot$  can be expressed in term of the $\hat{R}$ operators:
\beq
\R=\{\hat{\gamma}+\hat{R}\hat{\gamma},\hat{\gamma}\in\G\} \textrm{ and } \R^\bot=\{\hat{\gamma}-\hat{R}^*\hat{\gamma},\hat{\gamma}\in\G\}
.\eeq
Mimicking the same steps as previously, where we vary $h(\si,\tau)$ in the dual ansatz $l(\si,\tau)=\hat{g}(\si,\tau)h(\si,\tau) \in D$ with $\hat{g}(\si,\tau)\in\hat{G}$ and $h(\si,\tau)\in G$, we thus obtain the dual action on $\hat{G}$:
\beq
S(\hat{g})=\frac{1}{2}\int d\xi_+d\xi_-\langle \partial_- \hat{g}\hat{g}^{-1}|\big(R^{-1}+\hat{\Pi}^{\hat{g}}\big)^{-1}\partial_+ \hat{g}\hat{g}^{-1}\rangle,\quad \hat{g} \in \hat{G},
\eeq
where $\hat{\Pi}^g:\G\to\hat{\G}$ the Poisson-Lie operator on $\hat{\G}$.\medskip \\

\noindent We point out already in this section, that more general models appear when the kernel of the operator $R$ or/and the kernel of the sum $R+R^*$ is/are not trivial. In the Section 4 and 5, we will consider all possibilities just mentioned, and relate them to models already considered in Refs.\cite{KS1}, and \cite{Sfet1}. We note that the case where both kernels are not trivial has not been considered before and will constitute the principal original result of this paper.

\section{Models with $\Ker R=\Ker R^{*}=\{0\}$ and $\Ker (R+R^*)\ne\{0\}$}
In this part we suppose that $\Ker(R+R^*)\neq \{0\}$ while still $\Ker R=\Ker R^*=\{0\}$ and we show that this generalization leads again to a duality between two second order $\si$-models living on different target spaces. Moreover, under some further conditions on the operator $R$, the second order $\si$-models develop an intriguing gauge symmetry which is a novel aspect comparing to the case $\Ker(R+R^*)=\{0\}$.
\medskip
\\
As before the main structural ingredients of our story are the two operators $R:\hat{\G}\to\G$, $R^*:\hat{\G}\to\G$. Define $\hat{R}:\mathcal{G}\to\hat{\G}$ by the relation:
$$ \hat{R}=R^{-1}.$$
It follows that the kernel $\Ker(\hat{R}+\hat{R}^*)$ can be simply expressed as:
\beq
\Ker(\hat{R}+\hat{R}^*)=R\Ker(R+R^*).
\label{noyaux1}
\eeq
\paragraph{First order action:} Let us introduce the following linear subspaces of $\D$:
\beqn
\R&=&\{\gamma+R\gamma,\gamma\in\hat{\G}\} \label{Rspace}\\
\R^\bot&=&\{\gamma-R^*\gamma,\gamma\in\hat{\G}\} \label{Rbotspace}\\
\F&=&\{\gamma+R\gamma,\gamma\in\Ker (R+R^*)\}=\{\gamma-R^*\gamma,\gamma\in\Ker (R+R^*)\}.
\eeqn
Note that $\R^\bot$ is orthogonal to $\R$ as the notation suggests and $\F$ is necessarily isotropic ($\langle \F|\F\rangle=0$) since $\F=\R\cap\R^\bot$. We shall soon show (cf. Theorem 2) that the orthogonal complement $\F^\bot$ of $\F$ in $\D$ can be written as $\F^\bot=\R+\R^\bot$. We stress, however, that this is just the sum and not the direct sum of linear subspaces since $\R\cap\R^\bot\ne\{0\}$.\\
 Consider a $D$-valued field $l(\si,\tau)$ and an $\F$-valued field $\L(\si,\tau)$. Now write the first order action as:
\beq \S(l)=\frac{1}{2}\int \langle\partial_\tau ll^{-1}|\partial_\sigma ll^{-1}\rangle+\frac{1}{6}d^{-1}\langle dll^{-1}\wedge[dll^{-1}\wedge dll^{-1}]\rangle+ K(\partial_\sigma ll^{-1})+\langle \L|\partial_\sigma ll^{-1}\rangle.
\label{action3}
\eeq
Note that the field $\L$ is the Lagrange multiplier the variation of which imposes the constraint $\partial_\sigma ll^{-1}\in\F^\perp$.\\ 
Because of this new constraint, it is sufficient to define the quadratic form $K$ only on $\F^\bot$. For that we have to adapt the Theorem 1 to this new case.
\begin{theo} There exists a pair of linear operators $\alpha_\pm:\F^\bot\to\hat{\G}$ such that $X\in \F^\bot$ can be decomposed  as:
\beq
X=\alpha_+(X)+\alpha_-(X)+R\alpha_+(X)-R^*\alpha_-(X).
\label{decomposition2}
\eeq
\end{theo}
\begin{proof}
Let us consider the subspace of $\D$ formed by the sum $\R+\R^\bot$, and let us take any element $X\in \R+\R^\bot$. Then,  with the definitions \eqref{Rspace}\eqref{Rbotspace}, there exists couples $(\alpha, \beta)\in\hat{\G}$ such that:
\beq
X=\alpha+R\alpha+\beta-R^*\beta.
\eeq
Let us show that $\R+\R^\bot$ is orthogonal to $\F$, for that consider any element $\gamma \in \Ker(R+R^*)$ and calculate the following expression for $X\in \R+\R^\bot$:
\beqn
\langle X|\gamma+R\gamma\rangle&=&\langle \alpha+R\alpha+\beta-R^*\beta,\gamma+R\gamma \rangle \nonumber \\
&=&\langle\alpha+\beta|R\gamma\rangle+\langle R\alpha- R^*\beta|\gamma\rangle \nonumber \\
&=& \langle\alpha+\beta|R\gamma\rangle+\langle \alpha|R^*\gamma\rangle-\langle \beta|R\gamma\rangle \nonumber \\
&=&\langle \alpha|(R+R^*)\gamma\rangle \nonumber \\
&=&0. \nonumber \eeqn
The result is that $\R+\R^\bot \subset \F^\bot$. Let us analyze the dimension of the $\R+\R^\bot$ space, we have the equality:
\beqn
\dim(\R+\R^\bot)&=&\dim\R+\dim\R^\bot-\dim(\R\cap\R^\bot) \nonumber \\
&=&2n-\dim\F.
\label{dim1}
\eeqn
Denoting $\Ker(R+R^*)^0=\{x\in\G/\forall \gamma \in \Ker(R+R^*), \langle x|\gamma\rangle=0\}$ the annihilator of $\Ker(R+R^*)$ in $\G$, we consider any complement subspace $W$ of $\Ker(R+R^*)^0$ in $\G$, which means that we can write  $W\oplus\Ker(R+R^*)^0=\G$ and that $\dim W=\dim \F$. Let us introduce the following subspace $W+\F^\bot\subset \D$, and take any non null $v\in W$. Then since $\Ker(R+R^*)^0\cap W=\{0\}$  there exists a $\gamma\in \Ker(R+R^*)$ such that $\langle v|\gamma+R\gamma\rangle=\langle v |\gamma\rangle \ne 0$, therefore $W\cap\F^\bot=\{0\}$. Consequently, $W+\F^\bot=W\oplus\F^\bot$ and we obtain for the dimensions $\dim \F^\bot \le 2n-\dim \F$. Moreover, we have $\R+\R^\bot \subset \F^\bot$ which implies with \eqref{dim1} that $\dim \F^\bot=\dim(\R+\R^\bot)$ and the equality $ \F^\bot=\R+\R^\bot$.
\\
Therefore, for all $X\in\F^\bot$ there exists a (non-unique) couple $(\alpha, \beta)\in\hat{\G}$ such that $X=\alpha+R\alpha+\beta-R^*\beta$, and by identifying $\alpha\equiv\alpha_+(X)$ and $\beta\equiv\alpha_-(X)$ we obtain the proof.
\end{proof}
 \noindent  We can now define the quadratic form $K$ on $\F^\bot$ in terms of the operators $\alpha_\pm$:
\beq
K(X)=2\langle \alpha_+(X),R\alpha_+(X)\rangle+2\langle \alpha_-(X),R\alpha_-(X)\rangle,\; \forall X\in \F^\bot.
\label{quadratic_form}
\eeq
{\bf Remark:} We stress that the choice of $\alpha_\pm$ is not unique contrary to the case of the Section 2. Indeed, if we take any linear operator $\Phi:\F^\bot\to\Ker(R+R^*)$, the decomposition \eqref{decomposition2} remains unchanged by the tranformation $\alpha_\pm\to\alpha_\pm\pm\Phi$. Inspite of this ambiguity, we can check that the quadratic form $K$ \eqref{quadratic_form} is defined on $\F^\bot$ unambiguously, i.e. its value on $X\in \F^\bot$ does not depend on the choice of $\alpha_\pm$.
\medskip
\\
We can also easily determine the equations of motion verified by the $l$ field. By taking into account the constraint $\partial_\sigma ll^{-1}\in\mathcal{F}^\bot$, the extremalization  of the action \eqref{action3} leads to the result :
\beq
\partial_\tau ll^{-1} +\alpha_+(\partial_\sigma ll^{-1})+R\alpha_+(\partial_\sigma ll^{-1})-\alpha_-(\partial_\sigma ll^{-1})+R^*\alpha_-(\partial_\sigma ll^{-1})\in\F,
\label{motion}
\eeq
or written in the light-cone coordinates:
\beq
\alpha_+(\partial_+ ll^{-1})+R\alpha_+(\partial_+ ll^{-1})+\alpha_-(\partial_- ll^{-1})-R^*\alpha_-(\partial_- ll^{-1})\in\F.
\eeq
\paragraph{The second order action:}
In this paragraph we will present an algebraic derivation of the dual pair of the $\sigma$ models from the first order action \eqref{action3}. As before, since the double $D$ is perfect, we can represent any field $l(\sigma,\tau)\in D$ as $l(\sigma,\tau)=g(\sigma,\tau)\hat{h}(\sigma,\tau)$ with $g(\sigma,\tau)\in G$ and $\hat{h}(\sigma,\tau)\in\hat{G}$. With this parametrization, we can write the action \eqref{action3} as:
\beq
\S(g,\hat{h},\L)=\int \langle \partial_\tau gg^{-1}|\partial_\sigma l l^{-1}\rangle+\langle\partial_\sigma ll^{-1}|\Lambda \rangle +\frac{1}{2}K(\partial_\sigma l l^{-1}) .
 \label{action4}
\eeq
The constraint is solved by restricting the variable $\partial_\sigma l l^{-1}$ on $\F^\bot$, which, following the Theorem 2,  permits to perform the decomposition $\partial_\sigma ll^{-1}=\alpha_++\alpha_-+R\alpha_+-R^*\alpha_-$. The action then becomes simply:
\beqn
\mathcal{S}(g,\hat{h})=\int\langle \partial_\tau gg^{-1}| \alpha_++\alpha_-\rangle+\langle\alpha_+|R\alpha_+\rangle+\langle \alpha_-|R\alpha_- \rangle.
\eeqn
The elements $\alpha_\pm$ of $\hat{\G}$ are still related by the relation \eqref{relation1} as it can be seen by exactly the same reasoning as in Section 2:
\beq
(R+\Pi^g)\alpha_+-(R+\Pi^g)^*\alpha_-=\partial_\sigma gg^{-1}.
\eeq
And from the first order action \eqref{action3}, we derive the same action as Eq.\eqref{lag3bis} with $M^g=\Pi^g$ directly: 
\beq
S(g,\hat{h})=-\frac{1}{2}\int \langle Y|(R+\Pi^g)Y\rangle+\langle Z-\d_\tau gg^{-1}|(R+\Pi^g)^{-1}(Z+\d_\tau gg^{-1}\rangle,
\label{lag4bis}
\eeq
where
\beqn
Y&=&2\alpha_+-(R+\Pi^g)^{-1}\d_+ gg^{-1}
\label{relationY2}
\\
Z&=&(R+\Pi^g)\alpha_+-(R+\Pi^g)^*\alpha_-=\d_\si gg^{-1}.
\eeqn
The action \eqref{lag4bis} does not depend of the choice of the $\alpha_\pm$ operators. Indeed, if we make the transformation $\alpha_\pm \to \alpha_\pm\pm\Phi$ with $\Phi \in \Ker(R+R^*)$, the action \eqref{lag4bis} remains invariant. To see it, in a first time we observe that the quantity $Z$ is unchanged under that transformation. Secondly, under the same transformation the quantity $Y$ becomes $Y+2\Phi$, and the quadratic form $\langle Y|(R+\Pi^g)Y\rangle$ becomes:
\beqn
\langle Y+2\Phi|(R+\Pi^g)(Y+2\Phi)\rangle&=&\frac{1}{2}\langle Y+2\Phi|(R+R^*)(Y+2\Phi)\rangle\nonumber\\
&=&\frac{1}{2}\langle Y|(R+R^*)Y\rangle \nonumber\\
&=&\langle Y|(R+\Pi^g)Y\rangle.
\eeqn
The quadratic form $\langle Y|(R+\Pi^g)Y\rangle$ is therefore naturally defined on the coset $\hat{\G}/\Ker(R+R^*)$, which makes possible to write \eqref{lag4bis} as:
\beq
S(g,\hat{h})=-\frac{1}{2}\int \frac{1}{2}\langle \pi(Y)|(R+R^*)_\pi\pi(Y)\rangle-\langle \d_- gg^{-1}|(R+\Pi^g)^{-1}\d_+ gg^{-1}\rangle,
\label{lag4ter}
\eeq
where $\pi$ is the canonical projection $\pi:\hat{\G}\to\hat{\G}\slash \Ker(R+R^*)$, and the operator $(R+R^*)_\pi$ is nothing but the operator $(R+R^*)$ acting on the coset $\hat{\G}\slash \Ker(R+R^*)$.\\
The following Lemma 1 shows that varying $\hat{h}$ in $l=g\hat{h}$ while keeping fixed $g$ and respecting the constraint $\d_\si ll^{-1}\in \F^\bot$ is the same thing as varying $\pi(Y)$. This implies that we obtain from \eqref{lag4ter} directly the second order action:
\beq
S(g)=\frac{1}{2}\int \langle \d_- gg^{-1}|(R+\Pi^g)^{-1}\d_+ gg^{-1}\rangle,\quad g\in G.
\label{lag5}
\eeq
Indeed in \eqref{lag4ter}, there is no linear term in $\pi(Y)$, therefore varying with respect to $\pi(Y)$ amounts simply to the suppression of the quadratic term in $\pi(Y)$.
\begin{lemma} It holds that:
\beq
\pi(Y)=2(R+R^*)_\pi^{-1}\big((R+\Pi^g)^*\hat{J}\Ad_g\hat{J}\d_\si \hat{h}\hat{h}^{-1}+\d_\si gg^{-1}\big)-\pi\big((R+\Pi^g)^{-1}\d_+gg^{-1}\big).
\label{pi}
\eeq
\end{lemma}
\begin{proof}
Firstly, we have to show that the formula \eqref{pi} is well defined, which means that the quantity $V=(R+\Pi^g)^*\hat{J}\Ad_g\hat{J}\d_\si \hat{h}\hat{h}^{-1}+\d_\si gg^{-1}$ is in the image of $R+R^*$ whereas $\d_\si ll^{-1}\in \F^\bot$.\\
Let us consider any $\d_\si ll^{-1}\in \F^\bot$, which means that for all $\gamma\in\Ker(R+R^*)$:
\beq
\langle \Ad_g\d_\si \hat{h}\hat{h}^{-1}+\d_\si gg^{-1}|\gamma+R\gamma \rangle=0,
\eeq
hence
\beq
\langle (R+\Pi^g)^*\hat{J}\Ad_g\hat{J}\d_\si \hat{h}\hat{h}^{-1}+\d_\si gg^{-1}|\gamma\rangle=\langle V| \gamma \rangle=0 .
\eeq
Therefore $V\in \Ker(R+R^*)^0\subset\G$, it remains to prove that $\Img (R+R^*)=\Ker(R+R^*)^0$.\\
We can immediately see that $ \Img (R+R^*)\subset\Ker(R+R^*)^0$, indeed for $\gamma\in \Ker(R+R^*)$ we have $\langle \gamma|(R+R^*)X\rangle=\langle(R+R^*)\gamma|X\rangle=0$, where $X\in \hat{\G}$. Moreover, since the bilinear form $\langle.|.\rangle$ is non degenerate we have the relation $\dim \Ker(R+R^*)^0=n-\dim\Ker(R+R^*)$. From the rank-nullity theorem, we obtain that $\dim\Ker(R+R^*)+\dim\Img(R+R^*)=n$, together with the previous relation we conclude that $\dim \Ker(R+R^*)^0=\dim\Img(R+R^*)$ and $\Img (R+R^*)=\Ker(R+R^*)^0$. Consequently, $V\in \Img(R+R^*)$.\\
Secondly, let us prove the relation \eqref{pi}. From the decomposition of  $\d_\si ll^{-1}\in \F^\bot$ in the Theorem 2, we obtain:
\beq
(R+R^*)\alpha_+=(R^*\hat{J}\Ad_g\hat{J}+J\Ad_g\hat{J})\d_\si \hat{h}\hat{h}^{-1}+\d_\si gg^{-1}=V.
\label{image}
\eeq
The relation \eqref{image} can be expressed in term of the quantity $Y$ (cf Eq.\eqref{relationY2}):
\beq
(R+R^*)Y=2V
-(R+R^*)(R+\Pi^g)^{-1}\d_+gg^{-1}.
\label{image2}
\eeq
Hence:
\beq
(R+R^*)_\pi \pi(Y)=2V
-(R+R^*)_\pi\pi\big((R+\Pi^g)^{-1}\d_+gg^{-1}\big).
\label{image3}
\eeq
Since $(R+R^*)_\pi$ is an invertible operator on $\hat{\G}/\Ker(R+R^*)$, by inverting \eqref{image3} we obtain the desired result:
\beq
\pi(Y)=2(R+R^*)_\pi^{-1}V-\pi\big((R+\Pi^g)^{-1}\d_+gg^{-1}\big).
\eeq
\end{proof}
\noindent{\bf Remark:}
If we denote $p$ the dimension of $\Ker(R+R^*)$, then the $n-p$ variables contained in $\pi(Y)$ are related by the invertible relation \eqref{pi} to the $n-p$ independent components of $\d_\si \hat{h}\hat{h}^{-1}$ (Note that there are only $n-p$ independant components in $\d_\si \hat{h}\hat{h}^{-1}$, because $p$ components of $\d_\si \hat{h}\hat{h}^{-1}$ are determined by the relation $\d_\si ll^{-1}\in \F^\bot$).
\medskip

\noindent The dual second order action with respect to the action \eqref{lag5} is still obtain by taking the dual ansatz $l(\si,\tau)=\hat{g}(\si,\tau)h(\si,\tau) \in D$ with $\hat{g}(\si,\tau)\in\hat{G}$ and $h(\si,\tau)\in G$, and the well-known relation for the $\hat{R}:\G\to\hat{\G}$ operator:
\beq
\hat{R}=R^{-1}.
\eeq
The same reasoning leads to the dual second order action on $\hat{G}$:
\beq
S(\hat{g})=\frac{1}{2}\int \langle \partial_- \hat{g}\hat{g}^{-1}|\big(R^{-1}+\hat{\Pi}^{\hat{g}}\big)^{-1}\partial_+ \hat{g}\hat{g}^{-1}\rangle,\quad \hat{g} \in \hat{G}.
\label{lag6}
\eeq
\paragraph{The standard dressing cosets:} 
 Suppose that the linear subspace $\F$ is a Lie subalgebra of $\D$, whose Lie group $F$ is a connected Lie subgroup of $D$. Then:
\begin{theo} The first order action \eqref{action3} develops a gauge symmetry with $F$ as the gauge group if the operator $R$ verifies:
\beq
J\Ad_f\hat{J}+J[\Ad_f,R]\hat{J}=R\Ad_f R,\; \forall f\in F.
\label{condition}
\eeq
For $f(\si,\tau)\in F$, the revelant  gauge transformations read:
 $$l(\sigma,\tau)\to f(\sigma,\tau)l(\sigma,\tau), \quad \L(\si,\tau)\to  f(\sigma,\tau)\L(\sigma,\tau)f(\sigma,\tau)^{-1}-\partial_\tau f(\si,\tau)f(\sigma,\tau)^{-1}.$$
\end{theo}
\noindent {\bf Remark:} The condition \eqref{condition} is equivalent to the stability of the subspaces $\R$ and $\R^\bot$ by the adjoint action of $F$ as required in \cite{KS1}. This fact will show up in the proof of the Theorem 3. Furthermore, in Section 5, we will show that the previous hypothesis (isotropy of $\F$ and the condition \eqref{condition}) responsible for the gauge invariance of the first order action \eqref{action3}, will imply also a gauge invariance for the second order actions \eqref{lag5} and \eqref{lag6}. 
\begin{proof}
\noindent  We use the well-known Polyakov-Wiegmann formula \cite{PW}:
\beq
\S_0(fl)=\S_0(l)+\S_0(f)+\int \langle f^{-1}\partial_\tau f |  \partial_\sigma ll^{-1} \rangle,
\eeq
where
\beq
\S_0(l)=\frac{1}{2}\int \langle\partial_\tau ll^{-1}|\partial_\sigma ll^{-1}\rangle+\frac{1}{6}d^{-1}\langle dll^{-1}\wedge[dll^{-1}\wedge dll^{-1}]\rangle,
\eeq
and the formula 
\beq
\partial_\sigma ll^{-1} \to \partial_\sigma ff^{-1}+\Ad_f \partial_\sigma ll^{-1}.
\label{transformation}
\eeq
The isotropy of   $\F$ gives immediatly $\S_0(f)=0$, and with the constraint $\partial_\sigma ll^{-1}\in\F^\bot$ both ensure the gauge invariance of $\S_0(l)$.\\
Concerning gauge invariance of the term containing the quadratic form $K$ in the action \eqref{action3}, we can show that the transformation \eqref{transformation} leads to the following expression $K(\partial_\sigma ll^{-1})\to K( \partial_\sigma ff^{-1}+\Ad_f \partial_\sigma ll^{-1})$. By remarking that $\alpha_\pm( \partial_\sigma ff^{-1})\in\Ker(R+R^*)$, the expression of the quadratic form becomes simply $ K( \partial_\sigma ff^{-1}+\Ad_f \partial_\sigma ll^{-1})=K(\Ad_f \partial_\sigma ll^{-1})$. We will prove that $K(\Ad_f \partial_\sigma ll^{-1})=K(\partial_\sigma ll^{-1})$ in several steps.\\
Firstly, we will show that taking $X\in \F^\bot=\R+\R^\bot$ then $\Ad_f \F^\bot\subset \F^\bot$. Since $X\in \F^\bot$ it can be decomposed by using the relation \eqref{decomposition2} as $X=\alpha+R\alpha+\beta-R^*\beta$ with $(\alpha,\beta) \in \hat{\G}$. Thus if we act on $X\in\R+\R^\bot$ by the adjoint action of any element $f\in F$, we obtain on $\G$ and $\hat{\G}$:
\beqn
&\delta\equiv J\Ad_f\hat{J} \alpha+J\Ad_f J R \alpha& \in \G \\
&\hat{\delta}\equiv\hat{J}\Ad_f\hat{J} \alpha+\hat{J}\Ad_f J R \alpha& \in \hat{\G}.
\label{rela}
\eeqn
Moreover, the relation \eqref{rela} can be modified by using the condition \eqref{condition} on $R$ as:
\beqn
\hat{J}\Ad_f\hat{J} \alpha+\hat{J}\Ad_f J R \alpha&=&R\hat{J}\Ad_f\hat{J} \alpha+R\hat{J}\Ad_f J R \alpha \\
&=&R\hat{\delta}.
\eeqn
It implies that $\Ad_f \R\subset \R$. The same reasoning holds with the quantity $\beta$ which leads to $\Ad_f \R^\bot\subset \R^\bot$, and consequently $\Ad_f \F^\bot\subset \F^\bot$.\\
Secondly, since $\Ad_f X\in \F^\bot$ for $X\in \F^\bot$, we can use now the relation \eqref{decomposition2} to decompose it as:
\beqn
\Ad_f X&=&\Ad_f\alpha_+(X)+\Ad_f\alpha_-(X)+\Ad_fR\alpha_+(X)-\Ad_fR^*\alpha_-(X) \\
\Ad_f X&=& \alpha_+(\Ad_fX)+\alpha_-(\Ad_fX)+R\alpha_+(\Ad_fX)-R^*\alpha_-(\Ad_fX).
\eeqn 
Then, on the subspace $\R$ the equality between the previous relations leads to:
\beqn
J\Ad_f \alpha_+(X)+J\Ad_f R\alpha_+(X)&=&R\alpha_+(\Ad_fX) \label{eq1} \\
\hat{J}\Ad_f \alpha_+(X)+\hat{J}\Ad_f R\alpha_+(X)&=&\alpha_+(\Ad_fX).  \label{eq2}
\eeqn
Thirdly, let us take for example the first term of $K(\Ad_f X)$, and use the relations \eqref{eq1} and \eqref{eq2} to write:
\beqn
\langle \alpha_+(\Ad_f X)| R \alpha_+(\Ad_f X)\rangle&=&\langle R\alpha_+(X)|\alpha_+(X)\rangle+\langle \Xi_1^fR\alpha_+(X)|R\alpha_+(X)\rangle\nonumber \\
&&+\langle \Xi_2^f\alpha_+(X)|\alpha_+(X)\rangle.
\eeqn
Since $\Xi_1^f=\hat{J}\Ad_{f^{-1}}\hat{J}\Ad_fJ$ and $\Xi_2^f=J\Ad_{f^{-1}}\hat{J}\Ad_f\hat{J}$ are two skew-symetric operators in a symetric form, thus the last two terms are null. The case $\langle \alpha_-(\Ad_f X)| R \alpha_-(\Ad_f X)\rangle$ can be done in a similar way.
\end{proof}
 \medskip

\section{Models with $\Ker R=\Ker R^{*}\neq\{0\}$ and $\Ker(R+R^{*})\neq\{0\}$}
In Section 3, we demanded from the operators $R$ and $\hat{R}$ to have the trivial kernel; it is this last constraint which will be relaxed in that current section. We will see that, under some hypothesis on the operators $R$ and $\hat{R}$, the duality pair of non linear $\sigma$-models can be established even in this more general case.
\paragraph{Preliminaries:} We start with two operators $R:\Dom(R)\to\Img (R)$ and $\hat{R}:\Dom(\hat{R})\to\Img(\hat{R})$, where $\Dom R$ and $\Img(\hat{R})$ are subsets of $\hat{\G}$, and $\Dom \hat{R}$ and $\Img(R)$ are subsets of $\G$. Those operators verify the following hypothesis:
\begin{enumerate}
\item $\Ker R= \Ker R^*,\;\Ker \hat{R}= \Ker\hat{ R}^*$ 
\item $(R|_{\Img \hat{R}})^{-1}=\hat{R}|_{\Img R}$
\item $(\Ker \hat{R})^0=\Dom R$ and $(\Ker R)^0 = \Dom \hat{R}$
\item $\Img \hat{R}\subset\Dom R$ and $\Img R \subset \Dom \hat{R}$
\item $\Dom R=\Dom R^*$, $\Dom \hat{R}=\Dom \hat{R}^*$
\item The bilinear form $\langle.|.\rangle$ restricted to $\Img R\oplus \Img \hat{R}$ is non degenerate
\end{enumerate}
The first hypothesis is exactly the same as previously, the second one ensures the duality, the third one ensures the orthogonality of certain subspaces $\R$ and $\R^\bot$ of $\D$ that will generalize the subspaces $\R$ and $\R^{\bot}$ considered in Section 2 and 3, and the importance of the fifth shows up in the definition of the quadratic form $K$.\\
\noindent{\bf Remarks:} Firstly, the requirements 1-4 imply the following decompositions for the subspaces $\Dom R$ and $\Dom \hat{R}$:
\beq
\Dom R= \Img \hat{R}\oplus \Ker R, \quad \Dom \hat{R}=\Img R\oplus \Ker\hat{R}.
\eeq
Secondly, the relation \eqref{noyaux1} between the kernels of $R+R^*$ and $\hat{R}+\hat{R}^*$ is changed, indeed both operators $R$ and $\hat{R}$ have non trivial kernels, with the result that the link between the kernels becomes: \beq \Ker (\hat{R}+\hat{R}^*)|_{\Img R}=R|_{\Img\hat{R}}    \Ker (R+R^*)|_{\Img \hat{R}}.\eeq

In our new case, the linear subspaces $\R$, $\R^\bot$ and $\F$ are defined as follows:
\beqn
\R&=&\Ker \hat{R}\oplus\{\gamma+R\gamma,\gamma\in\Dom R\} \label{Rspace2}\\
\R^\bot&=&\Ker \hat{R}\oplus\{\gamma-R^*\gamma,\gamma\in\Dom R^*\} \label{Rbotspace2}\\
\F=\R\cap\R^\bot&=&\Ker \hat{R}\oplus\{\gamma+R\gamma,\gamma\in\Ker (R+R^*)\} \nonumber \\
&=&\Ker \hat{R}\oplus\{\gamma-R^*\gamma,\gamma\in\Ker (R+R^*)\}.
\eeqn
\noindent{\bf Remark:} The subspaces $\R$ and $\R^\bot$ can be rewritten in a duality invariant way:
\beqn
\R&=&\Ker \hat{R}\oplus\Ker R\oplus\{\gamma+R\gamma,\gamma\in\Img \hat{R}\} \\
&=&\Ker \hat{R}\oplus\Ker R\oplus\{\tilde{\gamma}+\hat{R}\tilde{\gamma},\tilde{\gamma}\in\Img R\}  \\
\R^\bot&=&\Ker \hat{R}\oplus\Ker R\oplus\{\gamma-R^*\gamma,\gamma\in\Img \hat{R}\} \\
&=&\Ker \hat{R}\oplus\Ker R\oplus\{\tilde{\gamma}-\hat{R}^*\tilde{\gamma},\tilde{\gamma}\in\Img R \} \\
\F&=&\Ker \hat{R}\oplus\Ker R\oplus\{\gamma+R\gamma,\gamma\in\Img \hat{R}\cap\Ker(R+R^{*})\} \\
&=&\Ker \hat{R}\oplus\Ker R\oplus\{\tilde{\gamma}+\hat{R}\tilde{\gamma},\tilde{\gamma}\in\Img R\cap\Ker(\hat{R}+\hat{R}^{*})\}
\eeqn
\paragraph{First order action:} The expression of the first order action is exactly the same as in Eq.\eqref{action3}, only the properties of the $R$ and $\hat{R}$ change. We still consider a $D$-valued field $l(\si,\tau)$ and an $\F$-valued field $\L(\si,\tau)$, and we write the first order action as:
\beq \S(l)=\frac{1}{2}\int \langle\partial_\tau ll^{-1}|\partial_\sigma ll^{-1}\rangle+\frac{1}{6}d^{-1}\langle dll^{-1}\wedge[dll^{-1}\wedge dll^{-1}]\rangle+ K(\partial_\sigma ll^{-1})+\langle \L|\partial_\sigma ll^{-1}\rangle .
\label{dress_first}
\eeq
Again, the field $\L$ is the Lagrange multiplier the variation of which imposes the constraint $\partial_\sigma ll^{-1}\in\F^\perp$, and in order to specify $K$ we need the following generalization of the Theorem 2:
\begin{theo} There exists a pair of linear operators $\alpha_\pm:\F^\bot\to\Dom R \subset\hat{\G}$ and an unique linear operator $\alpha_0:\F^\bot\to\Ker\hat{R}$ such that every $X\in\F^\bot$ can be decomposed as:
\beq
X=\alpha_+(X)+\alpha_-(X)+R\alpha_+(X)-R^*\alpha_-(X)+\alpha_0(X).
\label{decomposition3}
\eeq 
\end{theo}
\begin{proof}
The proof is similar to that of the Theorem 2, but adapted to the new definition of the $\F$ subspace.\\
 We can easily verify that $\R+\R^\bot\subset \F^\bot$, indeed  because of $\langle \Ker \hat{R}|\Dom R \rangle=0$ the two subspaces $\R$ and $\R^\bot$ remain orthogonal to $\F$.
Let us calculate the dimension of  $\R+\R^\bot$:
\beq
\dim  (\R+\R^\bot)=2\dim(\Dom R)+2\dim \Ker \hat{R}-\dim \F.
\eeq
The dimension of $\Dom R$ can be found from the relation $\Dom R=(\Ker\hat{R})^0$, which gives:
\beq
\dim (\Dom R)=n-\dim(\Ker \hat{R}),
\label{ani}
\eeq 
and implies:
\beq
\dim(\R+\R^\bot)=2n-\dim \F.
\eeq
Consider now two subspaces $\A$ and $\hat{\A}$ such that $\G=\Dom \hat{R}\oplus\hat{\A}$ and $\hat{\G}=\Dom R\oplus \A$. Thus it implies that $\dim \A=\dim \Ker \hat{R}$ and $\dim \hat{\A}=\dim \Ker R$. Futhermore, we denote $R|_{\Img \hat{R}}:\Img \hat{R}\to \Img R$ the restriction of $R$ on $\Img \hat{R}$. With this convention the linear subspace $\F$ can be written:
\beq
\F=\Ker R\oplus \Ker \hat{R}\oplus\{\gamma+R|_{\Img \hat{R}}\gamma,\; \gamma\in \Ker(R+R^{*})|_{\Img \hat{R}} \}.
\eeq
We choose any complement subspace $\U$ of $\Ker(R+R^{*})|_{\Img \hat{R}}^{0}$ in $\Img \hat{R}$, where $\Ker(R+R^{*})|_{\Img \hat{R}}^{0}$ is the annihilator of $\Ker(R+R^{*})|_{\Img \hat{R}}$ in $\Img \hat{R}$. Then we can write $\Img \hat{R}=\Ker(R+R^{*})|_{\Img \hat{R}}^{0}\oplus \U$, and $\Ker(R+R^{*})|_{\Img \hat{R}}^{0}\cap\U=\{0\}$. With this definition we directly obtain the dimension of the subspace $\U$:
\beq
\dim \U=\dim \Ker(R+R^{*})|_{\Img \hat{R}}.
\eeq
We consider the following subspace $(\A\oplus\hat{\A})+\U+\F^{\bot}$, and we shall prove that all the sums are direct.\\
We want to show that $\A\cap\F^{\bot}=\{0\}$. Suppose that there exists a non vanishing $a\in\A$ such that $ \langle a|\hat{k}\rangle=0$ for all $\hat{k}\in \Ker \hat{R}\subset \F$, then $a\in(\Ker \hat{R})^{0}=\Dom R$. Since $\A\cap\Dom R=\{0\}$, it implies that $a=0$ in contradiction with the hypothesis and we obtain that $\A\cap\F^{\bot}=\{0\}$.
The same reasoning holds for the subspace $\hat{A}$, leading to $\hat{A}\cap\F^{\bot}=\{0\}$.\\
Let us show now that $\U\cap \F^{\bot}=\{0\}$. Suppose that there exists a non vanishing $u\in\U$ such that $\langle u|\gamma\rangle=0$ for all $\gamma \in \Ker(R+R^{*})|_{\Img \hat{R}}$, then $u\in\Ker(R+R^{*})|_{\Img \hat{R}}^{0}$. But $\U\cap\Ker(R+R^{*})|_{\Img \hat{R}}^{0}=\{0\}$, it implies that $u=0$ contradicting the hypothesis. We obtain $\U\cap \F^{\bot}=\{0\}$.\\
Moreover $\U\cap(A\oplus\hat{A})=\{0\}$, because $\U\subset \Img\hat{R}$ and $\Img \hat{R}\cap A=\Img \hat{R}\cap\hat{A}=\{0\}$.\\
We can conclude that the following sum $A\oplus\hat{A}\oplus\U\oplus\F^{\bot}\subset \D$ is direct. Now observe that $\dim (A\oplus\hat{A}\oplus\U)=\dim \F$, since we have already shown that $\dim \U=\dim \Ker(R+R^{*})|_{\Img \hat{R}}$, $\dim \A=\dim \Ker \hat{R}$ and $\dim\hat{\A}=\dim \Ker R$. From the fact that $A\oplus\hat{A}\oplus\U\oplus\F^{\bot}\subset \D$, we obtain  the relation for the dimensions: $\dim \F^{\bot}+\dim\F\le 2n$.
\medskip \\
From relation \eqref{ani} we obtained $2n=\dim \F +\dim (\R+\R^{\bot})$, moreover we know that $\R+\R^{\bot}\subset\F^{\bot}$. Thus from the inequality $\dim \F^{\bot}+\dim\F\le 2n$ we observe that $\dim \F^\bot=2n-\dim\F=\dim  (\R+\R^\bot)$, which implies that $\R+\R^\bot=\F^\bot$.\\
Thus, for any $X\in \F^\bot=\R+\R^\bot$ we can decompose it as: $$X=\alpha_{+}(X)+R\alpha_{+}(X)+\alpha_{-}(X)-R^*\alpha_{-}(X)+\alpha_{0}(X), $$
with $\;(\alpha_{+},\alpha_{-})\in \Dom R,\; \alpha_{0} \in \Ker\hat{R}.$\\
We remark that the operator $\alpha_0$ is nothing but the identity on $\Ker\hat{R}$ and the null operator on the rest of $\F^\bot$.
\end{proof}
\noindent The quadratic form $K$ has the same expression as Eq.\eqref{quadratic_form}:
\beq
K(X)=2\langle \alpha_+(X),R\alpha_+(X)\rangle+2\langle \alpha_-(X),R\alpha_-(X)\rangle,\; \forall X\in \F^\bot,
\eeq
and $K$ is well defined on $\F^\bot$ since we required $\Dom R=\Dom R^*$.
\paragraph{The second order action:} Since the double $D$ is perfect, we parametrize any field $l(\sigma,\tau)\in D$ as $l(\sigma,\tau)=g(\sigma,\tau)\hat{h}(\sigma,\tau)$ with $g(\sigma,\tau)\in G$ and $\hat{h}(\sigma,\tau)\in\hat{G}$. We obtain the action:
\beq
\mathcal{S}(g,\hat{h},\L)=\int\langle \partial_\tau gg^{-1}|\partial_\sigma l l^{-1}\rangle+\langle\partial_\sigma ll^{-1}|\Lambda \rangle +\frac{1}{2}K(\partial_\sigma l l^{-1}),
\eeq
and:
\beq
\mathcal{S}(g,\hat{h},\L)=\int\langle \partial_\tau gg^{-1}|\alpha_++\alpha_-\rangle+\langle \alpha_+|R\alpha_+\rangle+\langle \alpha_-|R\alpha_-\rangle,
\label{action6}
\eeq
where we used in \eqref{action6} the decomposition $\partial_\sigma l l^{-1}=\alpha_++\alpha_-+R\alpha_+-R^*\alpha_-+\alpha_0$ on $\F^\bot$.\\
Furthermore, from the decomposition \eqref{decomposition3} considered for $X=\partial_{\sigma}ll^{-1}=\partial_{\sigma}gg^{-1}+\Ad_{g}\partial_{\sigma}\hat{h}\hat{h}^{-1}$, we obtain:
\beq
\partial_{\sigma} g g^{-1}=(R+\Pi^{g})\alpha_{+}-(R+\Pi^{g})\alpha_{-}+\alpha_{0}.
\label{inter}
\eeq
We can apply to the previous relation a projector $\rho$ from $\G$ to $\Img R\oplus \hat{\A}$ and such that $\Ker\rho=\Ker\hat{R}$, then the relation \eqref{inter} becomes:
\beq
\rho(\partial g g^{-1})=( R+\rho\circ\Pi^{g})\alpha_{+}- (R+ \rho\circ\Pi^{g})\alpha_{-},
\eeq
with the operator $ (R+ \rho\circ\Pi^{g}):\Dom R \to \Img R\oplus\hat{\A}$. Moreover, since $\langle \Dom R|\Ker \hat{R} \rangle=0$, the action \eqref{action6} can be rewritten:
\beq
\mathcal{S}(g,\hat{h},\L)=\int\langle \rho(\partial_\tau gg^{-1})|\alpha_++\alpha_-\rangle+\langle \alpha_+|R\alpha_+\rangle+\langle \alpha_-|R\alpha_-\rangle.
\label{action7}
\eeq
And, with the following relations,
\beqn
\tilde{Y}&=&2\alpha_{+}-(R+ \rho\circ\Pi^{g})^{-1}\rho(\partial_{+} gg^{-1})\nonumber \\
\tilde{Z}&=&( R+ \rho\circ\Pi^g)\alpha_+- (R+ \rho\circ\Pi^g)^*\alpha_-=\rho(\partial_\sigma gg^{-1}), \nonumber
\eeqn
the action \eqref{action7} can be written:
\beqn
S(g,\hat{h})&=&-\frac{1}{2}\int \langle \tilde{Y}|(R+ \rho\circ\Pi^g)\tilde{Y}\rangle\nonumber\\& +& \langle \tilde{Z}-\rho(\d_\tau gg^{-1})| (R+ \rho\circ\Pi^g)^{-1}(\tilde{Z}+\rho(\d_\tau gg^{-1})\rangle.
\label{lag6bis}
\eeqn
\\
Contrary to the previous case ($\Ker \hat{R}=\{0\}$), this time the variable $\tilde{Y}$ is defined on the $\Dom R\subset \hat{\G}$ and the variable $\tilde{Z}$ is defined on the subspace $\Img R\oplus \hat{\A}$ rather than on the whole $\G$ space.\\
The action \eqref{lag6bis} has exactly the same structure as the action \eqref{lag4bis} and leads, after the minimization, to the following second order action:
\beq
\mathcal{S}(g)=\frac{1}{2}\int d\xi^+d\xi^-\langle \rho(\partial_- gg^{-1})|  \big( R+ \rho\circ\Pi^g\big)^{-1}\rho(\partial_+ gg^{-1})\rangle,\quad g\in G .
\label{action9}
\eeq

\medskip

\noindent We won't do the calculus for the dual model, which is precisely the same as for the model on $G$, and leads to the following dual action:
\beq
\mathcal{S}(\hat{g})=\frac{1}{2}\int  d\xi^+d\xi^-\langle \tilde{\rho}(\partial_- \hat{g}\hat{g}^{-1})|     \big(\hat{R}+\tilde{\rho}\circ\hat{\Pi}^{\hat{g}}\big)^{-1}\tilde{\rho}(\partial_+ \hat{g}\hat{g}^{-1})\rangle,\quad \hat{g}\in \hat{G} ,
\label{action10}
\eeq
where $\tilde{\rho}$ is the projector from $\hat{\G}$ to the subspace $\Img \hat{R}\oplus \A$ with $\Ker\hat{\rho}=\Ker R$.
\medskip \\
 {\bf Remarks:} The duality relation $\hat{R}=R^{-1}$ of Section 2 and 3 has its analogue in the present section, indeed the operators $R$ and $\hat{R}$ are invertible on the respective subspaces $\Img \hat{R}$ and $\Img R$ and it holds $(R|_{\Img \hat{R}})^{-1}=\hat{R}|_{\Img R}$. 

\paragraph{The generalized dressing cosets}: We suppose again that the linear subspace $\F$ is a Lie subalgebra of $\D$, whose Lie group $F$ is a connected Lie subgroup of $D$. Then:
\begin{theo}
The first order action \eqref{dress_first} develops a gauge symmetriy with $F$ as the gauge group if both operators $R$ and $\hat{R}$ verify on their domains:
\beqn
J\Ad_f\hat{J}+J[\Ad_f,R]\hat{J}&=&R\Ad_f R, \label{condition_bis}\\
\hat{J}\Ad_fJ+\hat{J}[\Ad_f,\hat{R}]J&=&\hat{R}\Ad_f \hat{R},\; \forall f\in F.
\label{condition_ter}
\eeqn
\end{theo}
\noindent {\bf Remark:} Before the proof, let us make a short remark; because of the relation $$(R|_{\Img \hat{R}})^{-1}=\hat{R}|_{\Img R},$$ the relations \eqref{condition_bis},\eqref{condition_ter} are equivalent on the subspaces $\Img \hat{R}$ and $\Img R$ respectively. However, since $\Dom R=\Img \hat{R}\oplus \Ker R$ and $\Dom \hat{R}=\Img R\oplus \Ker \hat{R}$, the relations \eqref{condition_bis},\eqref{condition_ter} give new conditions for $R$ and $\hat{R}$ on the kernels.  If we choose $\delta\in\Ker \hat{R}$ and $\tilde{\delta}\in \Ker R$, then the relations \eqref{condition_bis},\eqref{condition_ter} give:
\beqn
\hat{J}\Ad_f J \delta&=&R\hat{J}\Ad_f J \tilde{\delta}, \\
J \Ad_f \hat{J} \tilde{\delta}&=&\hat{R}J\Ad_f \hat{J} \delta,
\eeqn 
implying that $\Ad_f \Ker \hat{R} \subset \Ker \hat{R}$ and $\Ad_f \Ker R \subset \Ker R$.
\begin{proof}
The proof is exactly the same as for the Theorem 3, since on one hand the action \eqref{dress_first} has the same structure as the first order action \eqref{action3} for the standard dressing cosets, and on the other hand the relation \eqref{condition} responsible of the gauge invariance of the action \eqref{action3} still exists and is generalized throught the relation \eqref{condition_bis} on the domain of $R$.
\end{proof}

\section{The gauge invariance of second order actions}
In this section we shall prove the gauge invariance of the second order actions of the types \eqref{lag5} and \eqref{action9} for all cases of dressing cosets considered in Section 2, 3 and 4. In other words, we shall etablish that the dual models live respectively on the target spaces $F\backslash G$ and $F\backslash \hat{G}$.\medskip
\\
We give now a proof of the gauge invariance of the second order actions of the type \eqref{lag5}, \eqref{action9} with respect to the dressing action of $F$ on $G$. Let us examine this dressing action by taking an element $f\in F$ and let it act on any element $l\in D$ by the standard left multiplication. Since $l=g\hat{h}$ with $g\in G$ and $\hat{h}\in \hat{G}$, the left action of $F$ on $D$ can be decomposed as \beq 
f l=(g\Delta g) (\Delta \hat{h}\hat{h}).
\label{transf}
\eeq
Then the dressing action on $G$ is given by $f\rhd g=g\Delta g$ and on $\hat{G}$ by $f\rhd \hat{h}=\Delta \hat{h}\hat{h}$. From Eq.\eqref{transf} we get for the variation $\Delta g$ and $\Delta \hat{h}$:
\beq
\Ad_{g^{-1}}f=\Delta g \Delta \hat{h}.
\eeq
In order to prove the gauge invariance we write the usual expression of the Poisson-Lie models in term of components. Recall that $\mathcal{R}$ is a $n$-dimensional linear subspace of $\mathcal{D}$ and can be written as a linear combinaison of the $2n$  generators $\{t_i,T^j\}$ of $\mathcal{D}$ such that $\langle t_i|T^j\rangle=\delta_i^j$, $\langle T^i|T^j \rangle=\langle t_i|t_j \rangle=0$ with $t_i\in \G$ and $T^i\in\hat{\G}$. Let us define the adjoint action of an element $g\in G$ on the generators $\{t,T\}$ of the Lie algebra $\mathcal{D}$:
\beq
\Ad_{g^{-1}}t_i=a(g)_i^jt_j,\quad\Ad_{g^{-1}}T^i=b(g)^{ij}t_j+a^{-1}(g)^i_jT^j,
\eeq
We can express the components of  $\Pi^g$ in term of the $a(g)$ and $b(g)$ matrices as $\Pi(g)^{ij}=b(g)^{ir}a^{-1}(g)^j_r$.
Then the action \eqref{lag5} becomes:
\beq 
\mathcal{S}(g)=\frac{1}{2}\int  d\xi^+d\xi^-(\partial_+ gg^{-1})^i(R+\Pi(g))^{-1}_{ij}(\partial_- gg^{-1})^j.
\label{DC}
\eeq
Let us consider the action \eqref{DC} evaluated in $g\Delta g$ for a given $f(\tau,\sigma)$ :
\beq
\S(g\Delta g)=\S(g)+\Delta \S(g)=\int d\xi^+d\xi^-( \mathcal{L}(g)+\Delta \mathcal{L}(g)),
\eeq
with $\mathcal{L}$ the Lagrangian of the action $S(g)$. To express the first order variation $\Delta \mathcal{L}(g)$ we need the infinitesimal variation of $\partial_\pm g g^{-1} \in\mathcal{G}$ and of the  Poisson-Lie bivector $\Pi(g)^{ij}$. Any element $\gamma\in\Ker( R+R^*)$ can be written as $\gamma=\gamma_\alpha T^\alpha$ where the  greek indices refer to $\Ker(R+R^*)=Span(T^{\alpha})$ with $\alpha=1,...,\dim \F$, and the subalgebra can be written $\F=Span(T^\alpha+R^{\alpha i}t_i)$. Thus an element $f$ of $F$ takes the form $f=e^{\epsilon_\alpha(T^\alpha+R^{\alpha i}t_i)}\in F$, where the  parameters $\epsilon_\alpha$ are the coordinates of $f$ on the group $F$.\\
The infinitesimal variations of the vectors $\partial_\pm gg^{-1}$ and of the bivector $\Pi(g)^{ij}$ are given by:
\beqn
&\Delta \partial_\pm gg^{-1}=\epsilon_\alpha \big[ R^{\alpha i}f_{li}^k-[\hat{f}^{\alpha k}_l-f^\alpha_{ls}\Pi(g)^{sk}]\big](\partial_\pm g g^{-1})^l t_k \label{rel1} \\
&\Delta \Pi(g)^{ij}=\epsilon_\alpha\big[R^{\alpha k}+\Pi(g)^{\alpha k}\big]\nabla_k\Pi(g)^{ij}, \label{rel2}
\eeqn
with $\nabla_k \Pi(g)^{ij}=-\tilde{f}^{ij}_k-f^i_{kr}\Pi(g)^{jr}+f^j_{kr}\Pi(g)^{ir}$, $f$ and $\hat{f}$ the structure constants of the Lie algebras $\G$ and $\hat{\G}$ respectively.\\
Denoting $E(g)=R+\Pi(g)$, the infinitesimal variation of the Lagrangian of the action \eqref{DC} is:
\beqn
\Delta \mathcal{L}(g)&=&(\Delta \partial_+ g g^{-1})^i(E(g)^{-1})_{ij}(\partial_- g g^{-1})^j+(\partial_+ g g^{-1})^i(E(g)^{-1})_{ij}(\Delta\partial_- g g^{-1})^j\nonumber \\ &-&(\partial_+ g g^{-1})^i(E(g)^{-1}\Delta E(g)E(g)^{-1})_{ij}(\partial_- g g^{-1})^j.
\label{varlag}
\eeqn
Replacing by equalities \eqref{rel1}-\eqref{rel2} and using the following Jacobi identity for $\Pi(g)$:
\beq
f^i_{sr}\Pi(g)^{js}\Pi(g)^{kr}+\tilde{f}^{ij}_s\Pi(g)^{ks}+cp(i,j,k)=0.
\eeq
Eq.\eqref{varlag} becomes:
\beq
\Delta \mathcal{L}(g)=\epsilon_\alpha(\partial_+ g g^{-1})^i(E(g)^{-1})_{ik}\Omega^{\alpha,kl}(E(g)^{-1})_{lj}(\partial_- g g^{-1})^j,
\eeq
where
$$\Omega^{\alpha,kl}=-R^{\alpha s}f^l_{st}R^{kt}+R^{ks}f^\alpha_{st}R^{tl}-R^{\alpha s}f^k_{st}R^{tl}+\hat{f}^{\alpha l}_sR^{ks}+\hat{f}^{kl}_sR^{\alpha s}+\hat{f}^{\alpha k}_sR^{sl}.$$\\
The vanishing of $\Omega^{\alpha,kl}$ is the consequence of the relation \eqref{condition}:
\beq
J\Ad_f\hat{J}+J[\Ad_f,R]\hat{J}=R\Ad_f R,\; \forall f\in F, \nonumber
\eeq
 for an infinitesimal variation $\delta f$.
We here thus proved the invariance of the action \eqref{lag5} under the dressing action of $F$ on $G$, i.e. $\mathcal{S}(g\Delta g)=\mathcal{S}(g)$, $g\in G$. Consequently, the target spaces of the dual pair are in fact the cosets $F\backslash G$ and $F\backslash \hat{G}$.\\
{\bf Remark 1:}
Eq.\eqref{condition} is not the only set of constraints on $R$, the following one comes from the fact that $\mathcal{F}$ is a Lie subalgebra therefore verifies $[\mathcal{F},\mathcal{F}]\subset \mathcal{F}$ or equivalently:
\beq
\iff \left \{ \begin{array}{l} \hat{f}_M^{\alpha\beta}+R^{\beta s}f^\alpha_{sM}-R^{\alpha s}f^\beta_{sM}=0 \\ \big[ R^{\alpha S}f^l_{st}R^{\beta t}-R^{\beta s}\hat{f}^{\alpha l}_s+R^{\alpha s}\hat{f}^{\beta l}_s\big](R^{-1})_{lM}=0
\end{array} \right. .
\label{null2}
\eeq
{\bf Remark 2:} In the case of the generalized dressing cosets the proof of the gauge invariance is similar. Indeed, if we denote $\hat{\A}\oplus\Img R=Span(t_{M})$ with $M=\dim \Ker \hat{R}+1,...,\dim \G$. Then the action \eqref{action9} becomes simply:
\beq \mathcal{S}(g)=\frac{1}{2}\int d\xi^+d\xi^-(\partial_+ gg^{-1})^M(R+\Pi_R(g))^{-1}_{MN}(\partial_- gg^{-1})^N.
\eeq
Since the relations \eqref{condition_bis} and \eqref{condition_ter} still hold and the generalized action has the same structure as the standard dressing cosets action, the generalized dressing cosets action is gauge invariant.

\section{Application of generalized dressing cosets: Models of Sfetsos \cite{Sfet1}}
Let us show that the models studied by Sfetsos \cite{Sfet1} enter in the category of the generalized dressing cosets. Recall from the beginning of Section 4 that the dual pair of the models is encoded in the choice of operators $R:\Dom R \to \Img R$ and $\hat{R}:\Dom\hat{R}\to\Img\hat{R}$. Here $\Dom \hat{R}$ and $\Img R$ are subsets of $\G$,  $\Dom R$ and $\Img \hat{R}$ are subsets of $\hat{\G}$, and both operators verify the relations \eqref{condition_bis} and \eqref{condition_ter}. We consider now a particular case which will turn out to give the Sfetsos models.
\begin{enumerate}
\item $\Ker R=\Ker (R+R^{*})=\{0\}$
\item $\Ker\hat{R}=\Ker(\hat{R}+\hat{R}^{*})\ne\{0\}$ and $\Ker\hat{R}$ is a Lie subalgebra of $\G$.
\end{enumerate}
Because the bilinear form is non degenerate on $\Img R\oplus\Img\hat{R}$, we can choose a basis $\{t_{i},T^{j}\}$ of $\D$ such that $\Img R=Span(t_{M}), \; \Img \hat{R}=Span(T^{M}),\; \Ker \hat{R}=Span(t_{\alpha})$, and \beq \langle t_{i}|T^{j}\rangle=\delta_{i}^{j},\quad (i,j)=1,...,n. \label{xy}\eeq
Here $\alpha=1,...,\dim \F$ and $M=\dim \F+1,..., \dim\G$.\\ Note in particular, that the condition \eqref{xy} implies also $\langle T^{M}|t_{N}\rangle=\delta^{M}_{N}$, $(M,N)=\dim \F+1,..., \dim\G$. We define the subspace $\A$ of $\hat{\G}$ as $\A=Span(T^{\alpha})$. We note that the points 1-2 imply that $\G=\Dom \hat{R}=\Img R\oplus \Ker \hat{R}$, $\hat{\G}=\Dom R=\Img \hat{R}\oplus \A$ and $\F=\Ker \hat{R}$.\\
 The action \eqref{action9} becomes that of Sfetsos \cite{Sfet1}:
\beq \mathcal{S}(g)=\frac{1}{2}\int d\xi^+d\xi^-(\partial_+ gg^{-1})^M(R+\Pi_R(g))^{-1}_{MN}(\partial_- gg^{-1})^N,
\label{action12}
\eeq
where $(O)_{{MN}}^{-1}$ means the inverse matrix of $\langle T^{M}|OT^{N} \rangle$. The dual theory is obtained by working out the action \eqref{action10} for our particular choice of $R$ and $\hat{R}$ which gives:
\beq
\mathcal{S}(\hat{g})=\frac{1}{2}\int  d\xi^+d\xi^-(\partial_+ \hat{g}\hat{g}^{-1})_i\left(\begin{array}{cc}(\hat{R}+ \hat{\Pi}_R(\hat{g}))_{\alpha \beta}&(\hat{R}+ \hat{\Pi}_R(\hat{g}))_{\alpha N}\\(\hat{R}+ \hat{\Pi}_R(\hat{g}))_{M \beta}& (\hat{R}+\hat{\Pi}_R(\hat{g}))_{MN}\end{array}\right)^{-1}(\partial_- \hat{g}\hat{g}^{-1})_j.
\eeq
Note that the components $\hat{R}_{\alpha \beta}$, $\hat{R}_{\alpha N}$ and $\hat{R}_{M \beta}$ vanish since $\Ker \hat{R}=Span(t_{\alpha})$, leading to:
\beq
\mathcal{S}(\hat{g})=\frac{1}{2}\int  d\xi^+d\xi^-(\partial_+ \hat{g}\hat{g}^{-1})_i\left(\begin{array}{cc}\hat{\Pi}_R(\hat{g})_{\alpha \beta}& \hat{\Pi}_R(\hat{g})_{\alpha N}\\ \hat{\Pi}_R(\hat{g})_{M \beta}& (\hat{R}+\hat{\Pi}_R(\hat{g}))_{MN}\end{array}\right)^{-1}(\partial_- \hat{g}\hat{g}^{-1})_j.
\label{action14}
\eeq
Furthermore, from $\hat{R}|_{\Img R}=(R|_{\Img \hat{R}})^{-1}$ we obtain $\hat{R}_{MN}=(R^{-1})_{MN}$, hence the action \eqref{action14} becomes that of the dual Sfetsos theory:
\beq
\mathcal{S}(\hat{g})=\frac{1}{2}\int  d\xi^+d\xi^-(\partial_+ \hat{g}\hat{g}^{-1})_i\left(\begin{array}{cc}\hat{\Pi}_R(\hat{g})_{\alpha \beta}& \hat{\Pi}_R(\hat{g})_{\alpha N}\\ \hat{\Pi}_R(\hat{g})_{M \beta}& \big(R^{-1}+\hat{\Pi}_R(\hat{g})\big)_{MN}\end{array}\right)^{-1}(\partial_- \hat{g}\hat{g}^{-1})_j.
\label{action15}
\eeq
Following the results of Section 5, Sfetsos'actions \eqref{action12} and \eqref{action15} are gauge invariant, if the conditions \eqref{condition_bis} and \eqref{condition_ter} are imposed. Note that the gauge invariance of the action \eqref{action12} was already proved by Sfetsos, however as far as the action \eqref{action15} is concerned, Sfetsos has only conjectured its gauge invariance. In our paper we have proved his conjecture with the group $F$ acting on $\hat{G}$ in the dressing way described in Eq.\eqref{transf}.
\section{Conclusions and Outlook}
We gave a new algebraic definition of the dual pair of $\si$-models introduced in \cite{KS1} under the name of the dressing cosets. Our new construction has led to a more general class of models than those constructed in \cite{KS1}. We call them the generalized dressing cosets. To specify more closely the character of the generalization let us make the book-keeping of all the models considered in this paper following the properties of the fundamental operators $R$ and $\hat{R}$.
\subparagraph{Standard Poisson-Lie T-duality \cite{KS2}:} $\Ker R= \Ker \hat{R}=\{0\}$ and $\Ker(R+R^*)=\{0\}$. There is no gauge group in this case.
\subparagraph{Standard dressing cosets \cite{KS1}:} $\Ker R= \Ker \hat{R}=\{0\}$ and $\Ker(R+R^*)\neq\{0\}$. The gauge group is a subgroup of $D$ but neither of $G$ nor of $\hat{G}$.
\subparagraph{Sfetsos models \cite{Sfet1}:} Either $\Ker R=\Ker(R+R^*)\neq\{0\}$ or $\Ker \hat{R}=\Ker(\hat{R}+\hat{R}^*)\neq\{0\}$. The gauge group is a subgroup of $G$ or of $\hat{G}$ respectively.
\subparagraph{New Sfetsos-like cases:} $\Ker R=\Ker(R+R^*)\neq\{0\}$ and $\Ker \hat{R}=\Ker(\hat{R}+\hat{R}^*)\neq\{0\}$. The gauge group is the direct product of the  Lie groups corresponding to the Lie algebras $\Ker R$ and $\Ker \hat{R}$.
\subparagraph{Generalized dressing cosets:}
$\Ker R\neq\{0\}$ and/or $\Ker \hat{R}\neq\{0\}$ and $\Ker(R+R^*)\neq\{0\}$. The gauge group is any kind of subgroup of $D$.
\medskip\\
Furthermore we have directly proved the gauge invariance of all second order actions. In the future we plan to study the dressing cosets in the case where the Drinfeld double $D$ is not perfect. 

\end{document}